%% file: NBJK-RTLS.tex
\documentclass[runningheads,a4paper]{llncs}
\pdfoutput=1

\input{NBJK-RTLS-preamble.tex}

\RequirePackage{hyperref}
\hypersetup{%
    pdfstartview=Fit,
    breaklinks=true,
    colorlinks=true
    }

\usepackage[strings]{underscore}

\usepackage[all]{hypcap}

\begin{document}
\title{Rewriting Theory for the Life Sciences:\\A Unifying Theory of CTMC Semantics\thanks{This is an extended version (containing additional technical appendices) of a paper with the same tittle accepted for \href{https://staf2020.hvl.no/events/icgt2020/}{ICGT~2020}.}}
\titlerunning{Rewriting Theory for the Life Sciences}
\author{Nicolas Behr\inst{1}\,\textsuperscript{\faEnvelopeO} 
\and Jean Krivine\inst{2}}
\authorrunning{N. Behr and J. Krivine}
%
\institute{Center for Research and Interdisciplinarity (CRI)\\
Universit\'{e} de Paris, INSERM U1284\\
8-10 Rue Charles V, 75004 Paris, France\\
\email{nicolas.behr@cri-paris.org}
\and
Institut de Recherche en Informatique Fondamentale (IRIF)\\
Universit\'{e} de Paris, CNRS UMR 8243\\
8 Place Aur\'{e}lie Nemours, 75205 Paris Cedex 13, France\\
\email{jean.krivine@irif.fr}}
\maketitle
\begin{abstract}
The Kappa biochemistry and the M\O{}D organo-chemistry frameworks are amongst the most intensely developed applications of rewriting theoretical methods in the life sciences to date. %
A typical feature of these types of rewriting theories is the necessity to implement certain structural constraints on the objects to be rewritten (a protein is empirically found to have a certain signature of sites, a carbon atom can form at most four bonds, \ldots). %
In this paper, we contribute to the theoretical foundations of these types of rewriting theory a number of conceptual and technical developments that permit to implement a universal theory of continuous-time Markov chains (CTMCs) for stochastic rewriting systems. Our core mathematical concepts are a novel rule algebra construction for the relevant setting of rewriting rules with conditions, both in Double- and in Sesqui-Pushout semantics, augmented by a suitable stochastic mechanics formalism extension that permits to derive dynamical evolution equations for pattern-counting statistics. %
\keywords{Double-Pushout rewriting \and Sesqui-pushout rewriting \and rule algebra %
\and stochastic mechanics \and biochemistry \and organic chemistry.}
\end{abstract}

\section{Motivation}

One of the key applications that rewriting theory may be considered for in the life sciences is the theory of continuous-time Markov chains (CTMCs) modeling complex systems. In fact, since Delbr\"{u}ck's seminal work on autocatalytic reaction systems in the 1940s~\cite{Delbr_ck_1940}, the mathematical theory of chemical reaction systems has effectively been formulated as a rewriting theory in disguise, namely via the rule algebra of discrete graph rewriting~\cite{bp2019-ext}. In the present paper, we provide the necessary technical constructions in order to consider the CTMCs and analysis methods of relevance for more general types of compositional rewriting theories with conditions, with key examples provided in the form of  \emph{biochemical graph rewriting} in the sense of the \KAP{} framework (\url{https://kappalanguage.org})~\cite{Boutillier:2018aa}, and \emph{(organo-) chemical graph rewriting} in the sense of the \MOD{} framework (\url{https://cheminf.imada.sdu.dk/mod/})~\cite{Andersen_2016}. The present paper aims to serve two main purposes: the first consists in providing an extension of the existing category-theoretical rule-algebra frameworks~\cite{bp2018,bp2019-ext,nbSqPO2019} by the rewriting theoretical design feature of incorporating rules with conditions as well as constraints on objects (Section~\ref{sec:ra}). Based upon these technical developments, we then investigate to which extent it is possible to utilize the rule-algebraic \emph{stochastic mechanics frameworks} of the relevant types (Section~\ref{sec:stochMech}) in order to compute evolution equations for the moments of pattern-count observables within the \KAP{} and \MOD{} frameworks (Section~\ref{sec:bcgr} and~\ref{sec:ocgr}).

\section{Compositional rewriting theories with conditions}

The well-established \emph{Double-Pushout (DPO)}~\cite{ehrig:2006fund} and \emph{Sesqui-Pushout (SqPO)}~\cite{Corradini_2006} frameworks for rewriting systems over categories with suitable adhesivity properties~\cite{lack2005adhesive,ehrig2004adhesive,GABRIEL_2014,ehrig2014mathcal} provide a principled and very general  foundation for rewriting theories. In practice, many applications require the rewriting of objects that are not part of an adhesive category themselves, but which may be obtained from a suitable ``ambient'' category via the notion of \emph{conditions} on objects. Together with a corresponding notion of constraints on rewriting rules, this yields a versatile extension of rewriting theory. In the DPO setting, this modification had been well-known~\cite{habel2009correctness,ehrig:2006fund,ehrig2014mathcal,ehrig2012m}, while it has been only very recently introduced for the SqPO setting~\cite{behrRaSiR}. For the \emph{rule algebra} constructions presented in the main part of this contribution, we require in addition a certain \emph{compositionality} property of our rewriting theories (established for the DPO case in~\cite{bp2018,bp2019-ext}, for the SqPO case in~\cite{nbSqPO2019}, and for both settings augmented with conditions in~\cite{behrRaSiR}). 

\subsection{Category-theoretical prerequisites}

We collect in Appendix~\ref{sec:MACapp} some of the salient concepts on $\cM$-adhesive categories and the relevant notational conventions. Throughout this paper, we will make the following assumptions:
\begin{assumption}\label{as:main}
    $\bfC\equiv(\bfC,\cM)$ is a finitary $\cM$-adhesive category with $\cM$-initial object, $\cM$-effective unions and epi-$\cM$-factorization. In the setting of \emph{Sesqui-Pushout (SqPO) rewriting}, we assume in addition that all final pullback complements (FPCs) along composable pairs of $\cM$-morphisms exist, and that $\cM$-morphisms are stable under FPCs.
\end{assumption}

Both of the main application examples presented within this paper rely upon typed variants of undirected multigraphs. 
\begin{definition}
    Let $\cP^{(1,2)}:\mathbf{Set}\rightarrow \mathbf{Set}$ be the restricted powerset functor (mapping a set $S$ to the set of its subsets $P\subset S$ with $1\leq |P|\leq 2$). The category $\mathbf{uGraph}$~\cite{bp2019-ext} of \emph{finite undirected multigraphs} is defined as the finitary restriction of the comma category $(ID_{\mathbf{Set}},\cP^{(1,2)})$. Thus an undirected multigraph is specified as $G=(E_G,V_G,i_G)$, where $E_G$ and $V_G$ are (finite) sets of edges and vertices, respectively, and where $i_G:E_G\rightarrow \cP^{(1,2)}(V_G)$ is the edge-incidence map.
\end{definition}

\begin{theorem}
    $\mathbf{uGraph}$ satisfies Assumption~\ref{as:main}, both for the DPO- and for the extended SqPO-variant.
\end{theorem}
\begin{proof}
	As demonstrated in~\cite{bp2019-ext}, $\mathbf{uGraph}$ is indeed a finitary $\cM$-adhesive category with $\cM$-initial object and $\cM$-effective unions, for $\cM$ the class of component-wise monic $\mathbf{uGraph}$-morphisms. It thus remains to prove the existence of an epi-$\cM$-factorization as well as the properties related to FPCs. To this end, utilizing the fact that the category $\mathbf{Set}$ upon which the comma category $\mathbf{uGraph}$ is based possesses an epi-mono-factorization, we may construct the following diagram from a $\mathbf{uGraph}$-morphism $\varphi=(\varphi_E,\cP^{(1,2)}(\varphi_V))$ (for component morphisms $\varphi_E:E\rightarrow E'$ and $\varphi_V:V\rightarrow V'$):
\begin{equation}
    \inputtikz{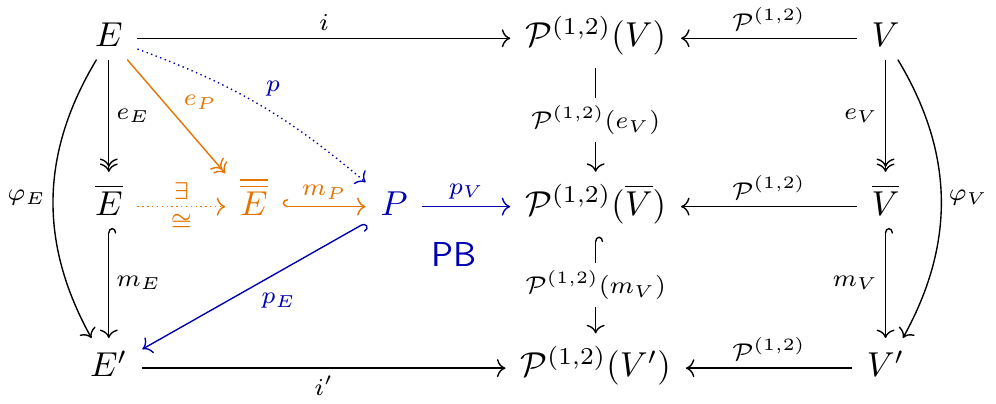}
\end{equation}
The diagram is constructed as follows:
\begin{enumerate}
\item Perform the epi-mono-factorizations $\varphi_E=m_E\circ e_E$ and $\varphi_V=m_V\circ e_V$, and apply the functor $\cP^{(1,2)}$ in order to obtain the morphisms $\cP^{(1,2)}(e_V)$ and $\cP^{(1,2)}(m_V)$; since  the functor $\cP^{(1,2)}$ preserves monomorphisms~\cite{padberg2017towards}, $\cP^{(1,2)}(m_V)\in \mono{\mathbf{Set}}$.
\item Construct the pullback 
\[(E'{\color{h1color}\leftarrow P\rightarrow}\cP^{(1,2)}(\overline{V})):=\pB{E'\rightarrow \cP^{(1,2)}(V')\leftarrow \cP^{(1,2)}(\overline{V})}\,,
\]
Since monomorphisms are stable under pullback in $\mathbf{Set}$, having proved that $\cP^{(1,2)}(m_V)\in \mono{\mathbf{Set}}$ implies ${\color{h1color}(p_E:P\rightarrow E')}\in \mono{\mathbf{Set}}$.
\item By the universal property of pullbacks, there exists a morphism ${\color{h1color}(p:E\rightarrow P)}$. Let $p={\color{h2color}m_P\circ e_P}$ be the epi-mono-factorization of this morphism. 
\item By stability of monomorphisms under composition in $\mathbf{Set}$, we find that ${\color{h1color}p_E}\circ {\color{h2color}m_P}\in \mono{\mathbf{Set}}$, and consequently $\varphi_E=({\color{h1color}p_E}\circ {\color{h2color}m_P})\circ {\color{h2color}e_P}$ yields an alternative epi-mono-factorization of $\varphi_E$. Then by uniqueness of epi-mono-factorizations up to isomorphism, there must exist an isomorphism $(\overline{E}{\color{h2color}\rightarrow \overline{\overline{E}}})\in \iso{\mathbf{Set}}$.
\end{enumerate}
We have thus demonstrated that both $(e_E,\cP^{(1,2)}(e_V))$ and $(m_E,\cP^{(1,2)}(m_V))$ are morphisms in $\mathbf{uGraph}$. Since morphisms in comma categories are mono-, epi- or iso-morphisms if they are so componentwise~\cite{ehrig:2006fund}, we conclude that 
\[
(e_E,\cP^{(1,2)}(e_V))\in\epi{\mathbf{uGraph}}\,,\quad (m_E,\cP^{(1,2)}(m_V))\in\mono{\mathbf{uGraph}}\,,
\]
which finally entails that we have explicitly constructed an epi-mono-factorization of the $\mathbf{uGraph}$-morphism $(\varphi_E,\cP^{(1,2)}(\varphi_V))$. 

In order to demonstrate that FPCs along pairs of composable $\cM$-morphisms $\varphi_A,\varphi_B\in\cM$ in $\mathbf{uGraph}$ exist (for $\cM$ the class of component-wise monomomophic $\mathbf{uGraph}$ morphisms), we provide the following explicit construction:
\begin{equation}
\begin{array}{c|c}
    \inputtikz{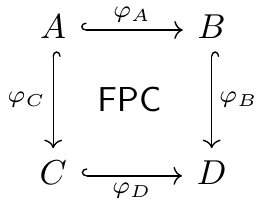}\hphantom{X} & \hphantom{X}
    \begin{aligned}
        V_C&=V_D\setminus(V_B\setminus V_A)\\
        E_C&=\{e\in E_D\setminus (E_B\setminus E_A)\mid
        u_D(e)\in\cP^{(1,2)}(V_C)\}\\
        u_C&=u_D\vert_{E_C}\\
        \varphi_C&=(E_A\hookrightarrow E_C,\cP^{(1,2)}(V_A\hookrightarrow V_C))\\
        \varphi_D&=(E_C\hookrightarrow E_D,\cP^{(1,2)}(V_C\hookrightarrow V_D))
    \end{aligned}
    \end{array}
\end{equation}
\end{proof}

\subsection{Conditions}\label{sec:cond}

Referring to Appendix~\ref{sec:condApp} for further details and technical definitions, we will utilize as a \textbf{notational convention} the standard shorthand notations
\begin{equation}
    \exists(X\hookrightarrow Y):=\exists(X\hookrightarrow Y,\ac{true}_Y)\,,\quad
    \forall(X\hookrightarrow Y,\ac{c}_Y):=\neg\exists(X\hookrightarrow Y,\neg\ac{c}_Y)\,.
\end{equation}
For example, the constraints
\begin{equation*}
    \ac{c}_{\mIO}^{(1)}=\exists(\mIO\hookrightarrow \inputtikz{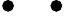})\,,\; \ac{c}_{\mIO}^{(2)}=\not \exists(\mIO\hookrightarrow \inputtikz{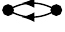})\,,\;
\ac{c}_{\mIO}^{(3)}=\forall(\mIO\hookrightarrow 
\inputtikz{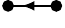},\exists(
\inputtikz{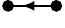}\hookrightarrow 
\inputtikz{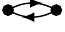}))
\end{equation*}
express for a given object $Z\in\obj{\bfC}$ that $Z$ contains at least two vertices (if $Z\vDash\ac{c}_{\mIO}^{(1)}$), that $Z$ does not contain parallel pairs of directed edges (if $Z\vDash\ac{c}_{\mIO}^{(2)}$), and that for every directed edge in $Z$ there also exists a directed edge between the same endpoints with opposite direction (if $Z\vDash\ac{c}_{\mIO}^{(3)}$), respectively.

\subsection{Compositional rewriting with conditions}\label{sec:crc}

Throughout this section, we assume that we are given a type $\bT\in\{DPO,SqPO\}$ of rewriting semantics and an $\cM$-adhesive category $\bfC$ satisfying the respective variant of Assumption~\ref{as:main}. In categorical rewriting theories, the universal constructions utilized such as pushouts, pullbacks, pushout complements and final pullback complements are unique only up to universal isomorphisms. This motivates specifying a suitable notion of equivalence classes of rules with conditions:

\begin{definition}[Rules with conditions]\label{def:rwc}
    Let $\LinAc{\bfC}$ denote the class of \emph{(linear) rules with conditions}, defined as
    \begin{equation}
        \LinAc{\bfC}:=\{(O\xleftarrow{o}K\xrightarrow{i}I;\ac{c}_I)\mid o,i\in\cM,\; \ac{c}_I\in\cond{\bfC}\}\,.
    \end{equation}
{\makeatletter
\let\par\@@par
\par\parshape0
\everypar{}
\begin{wrapfigure}[5]{r}{0.44\linewidth}
\vspace{-0.9em}
  \begin{equation}
       \inputtikz{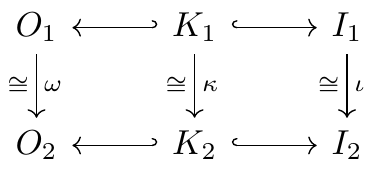}
  \end{equation}
\end{wrapfigure}
\noindent We define two rules with conditions $R_j=(r_j,\ac{c}_{I_j})$ ($j=1,2$) \emph{equivalent}, denoted $R_2\sim R_1$, iff $\ac{c}_{I_1}\equiv\ac{c}_{I_2}$ and if there exist isomorphisms $\omega,\kappa,\iota\in\iso{\bfC}$ such that the diagram on the right commutes. We denote by $\LinEq{\bfC}$ the set of equivalence classes under $\sim$ of rules with conditions.\par}
\end{definition}

\begin{definition}[Direct derivations]
    Let $r=(O\hookleftarrow K\hookrightarrow I)\in\Lin{\bfC}$ and $\ac{c}_I\in\cond{\bfC}$ be concrete representatives of some equivalence class $R\in\LinEq{\bfC}$, and let $X,Y\in\obj{\bfC}$ be objects. Then a \emph{type $\bT$ direct derivation} is defined as a commutative diagram such as below right, where all morphism are in $\cM$ (and with the left representation a shorthand notation)
\begin{equation}\label{eq:DD}
\inputtikz{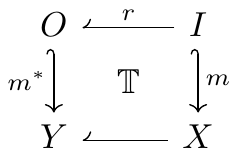}\quad :=\quad 
\inputtikz{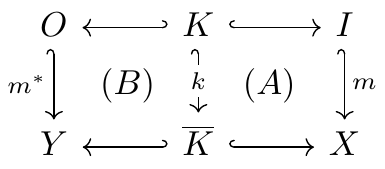}\,.
\end{equation}
with the following pieces of information required relative to the type:
\begin{enumerate}
    \item $\mathbf{\bT=DPO}$: given $(m:I\hookrightarrow X)\in\cM$, $m$ is a \emph{DPO-admissible match of $R$ into $X$}, denoted $m\in\MatchGT{DPO}{R}{X}$, if $m\vDash \ac{c}_I$ and $(A)$ is constructable as a \emph{pushout complement}, in which case $(B)$ is constructed as a \emph{pushout}.
    \item $\mathbf{\bT=SqPO}$: given $(m:I\hookrightarrow X)\in\cM$, $m$ is a \emph{SqPO-admissible match of $R$ into $X$}, denoted $m\in\MatchGT{SqPO}{R}{X}$, if $m\vDash\ac{c}_I$, in which case $(A)$ is constructed as a \emph{final pullback complement} and $(B)$ as a \emph{pushout}.
    \item $\mathbf{\bT=DPO^{\dag}}$: given just the ``plain rule'' $r$ and $(m^{*}:O\hookrightarrow Y)\in\cM$, $m^{*}$ is a \emph{DPO${}^{\dag}$-admissible match of $r$ into $X$}, denoted $m\in\MatchGT{DPO^{\dag}}{r}{Y}$, if $(B)$ is constructable as a \emph{pushout complement}, in which case $(B)$ is constructed as a \emph{pushout}.
\end{enumerate}
For types $\bT\in\{DPO,SqPO\}$, we will sometimes employ the notation $R_m(X)$ for the object $Y$.
\end{definition}
Note that at this point, we have not yet resolved a conceptual issue that arises from the non-uniqueness of a direct derivation given a rule and an admissible match. Concretely, the pushout complement, pushout and FPC constructions are only unique up to isomorphisms. This issue will ultimately be resolved as part of the rule algebraic theory. We next consider a certain \emph{composition operation} on rules with conditions that is quintessential to our main constructions:

\begin{definition}[Rule compositions]\label{def:Rcomp}
 Let $R_1,R_2\in \LinEq{\bfC}$ be two equivalence classes of rules with conditions, and let $r_j\in\Lin{\bfC}$ and $\ac{c}_{I_j}$ be concrete representatives of $R_j$ (for $j=1,2$). For $\bT\in\{DPO,SqPO\}$, an $\cM$-span $\mu=(I_2\hookleftarrow M_{21}\hookrightarrow O_1)$ (i.e.\ with $(M_{21}\hookrightarrow O_1),(M_{21}\hookrightarrow I_2)\in\cM$) is a \emph{$\bT$-admissible match of $R_2$ into $R_1$} if the diagram below is constructable (with $N_{21}$ constructed by taking pushout)
 \begin{equation}\label{eq:defRcomp}
 \inputtikz{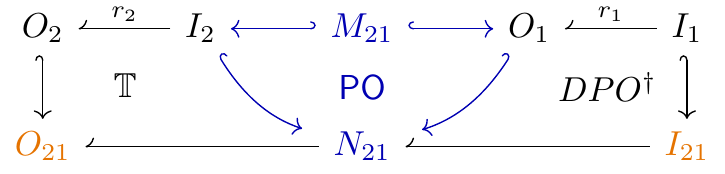}
\end{equation}
and if $\ac{c}_{I_{21}}\not{\!\!\dot{\equiv}}\,\,\ac{false}$. Here, the condition $\ac{c}_{I_{21}}$ is computed as
\begin{equation}
    \ac{c}_{I_{21}}:=\Shift(I_1\hookrightarrow I_{21},\ac{c}_{I_1})\;\land\; \Trans(N_{21}\leftharpoonup I_{21},\Shift(I_2\hookrightarrow N_{21},\ac{c}_{I_2}))\,.
\end{equation}
In this case, we define the \emph{type $\bT$ composition of $R_2$ with $R_1$ along $\mu$}, denoted $\compGT{\bT}{R_2}{\mu}{R_1}$, as
\begin{equation}
\compGT{\bT}{R_2}{\mu}{R_1}:=[(O_{21}\leftharpoonup I_{21};\ac{c}_{I_{21}})]_{\sim}\,,
\end{equation}
where $(O_{21}\leftharpoonup I_{21}):=(O_{21}\leftharpoonup N_{21})\circ(N_{21}\leftharpoonup I_{21})$ (with $\circ$ the \emph{span composition} operation).
\end{definition}

We recall in Appendix~\ref{app:ACthms} two important technical results on the notions of direct derivations and rule compositions that have been derived in~\cite{behrRaSiR} (where however the DPO-type concurrency theorem is of course classical, cf.\ e.g.\ \cite{ehrig:2006fund}).

\section{Rule algebras for compositional rewriting with conditions}\label{sec:ra}

The associativity property of rule compositions in both DPO- and SqPO-type semantics for rewriting with conditions as proved in~\cite{behrRaSiR} may be fruitfully exploited within rule algebra theory. One possibility to encode the non-determinism in sequential applications of rules to objects is given by lifting each possible configuration $X\in\obj{\bfC}_{\cong}$ (i.e.\ isomorphism class of \emph{objects}) to a basis vector $\ket{X}$ of a vector space $\hat{\bfC}$; then a rule $r$ is lifted to a linear operator acting on $\hat{\bfC}$, with the idea that this operator acting upon a basis vector $\ket{X}$ should evaluate to the ``sum over all possibilities to act with $r$ on $X$''. We will extend here the general rule algebra framework~\cite{bdg2016,bp2018,nbSqPO2019} to the present setting of rewriting rules with conditions.\\

We will first lift the notion of rule composition into the setting of a composition operation on a certain abstract vector space over rules, thus realizing the heuristic concept of ``summing over all possibilities to compose rules''.
\begin{definition}
    Let $\bT\in\{DPO,SqPO\}$ be the rewriting type, and let $\bfC$ be a category satisfying the relevant variant of Assumption~\ref{as:main}. Let $\overline{\cR}_{\bfC}$ be an $\bR$-vector space, defined via a bijection $\delta:\LinEq{\bfC}\xrightarrow{\cong}\mathsf{basis}(\overline{\cR}_{\bfC})$ from the set of equivalence classes of linear rules with conditions to the set of basis vectors of $\overline{\cR}_{\bfC}$. Let $\rap{\bT}{}{}$ denote the \emph{type $\bT$ rule algebra product}, a binary operation defined  via its action on basis elements $\delta(R_1),\delta(R_1)\in \overline{\cR}_{\bfC}$ (for $R_1,R_2\in \LinEq{\bfC}$) as
    \begin{equation}
        \rap{\bT}{\delta(R_2)}{\delta(R_1)}:=\sum_{\mu\in \RMatchGT{\bT}{R_2}{R_1}}\delta\left(\compGT{\bT}{R_2}{\mu}{R_1}\right)\,.
    \end{equation}
    We refer to $\overline{\cR}_{\bfC}^{\bT}:=(\overline{\cR}_{\bfC},\rap{\bT}{}{})$ as the \emph{$\bT$-type rule algebra over $\bfC$}.
\end{definition}
\begin{theorem}
    For type $\bT\in\{DPO,SqPO\}$ over a category $\bfC$ satisfying Assumption~\ref{as:main}, the rule algebra $\overline{\cR}_{\bfC}^{\bT}$ is an \emph{associative unital algebra}, with unit element $\delta(R_{\mIO})$, where $R_{\mIO}:=(\mIO\hookleftarrow\mIO\hookrightarrow \mIO;\ac{true})$.
\end{theorem}
\begin{proof}
    \emph{Associativity} follows from Theorem~\ref{thm:assocR}, while \emph{unitality}, i.e.\ that 
    \[
      \forall R\in \LinEq{\bfC}:\quad  \rap{\bT}{\delta(R_{\mIO})}{\delta(R)}
      =\rap{\bT}{\delta(R)}{\delta(R_{\mIO})}=\delta(R)
    \]
    follows directly from an explicit computation of the relevant rule compositions.
\end{proof}

As alluded to in the introduction, the prototypical example of rule algebras are those of DPO- or (in this case equivalently) SqPO-type over discrete graphs, giving rise as a special case to the famous Heisenberg-Weyl algebra of key importance in mathematical chemistry, combinatorics and quantum physics (see~\cite{bp2019-ext} for further details). We will now illustrate the rule algebra concept in an example involving a more general base category.

\begin{example}
    For the category $\mathbf{uGraph}$ and DPO-type rewriting semantics, consider as an example the following two rules with conditions:
    \begin{equation}
        R_C:=\left(\inputtikz{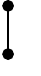}
	\hookleftarrow \inputtikz{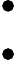}
	\hookrightarrow \inputtikz{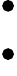};\neg \exists \left(\inputtikz{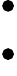}\hookrightarrow \inputtikz{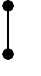}\right)\right)\,,\quad R_V:=(\inputtikz{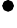}\hookleftarrow\mIO\hookrightarrow\inputtikz{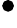};\ac{true})\,.
    \end{equation}
The first rule is a typical example of a rule with application conditions, i.e.\ here stating that the rule may only link two vertices if they were previously not already linked to each other. The second rule, owing to DPO semantics, can in effect only be applied to vertices without any incident edges. The utility of the rule-algebraic composition operation then consists in reasoning about sequential compositions of these rules, for example (letting $*:=\rap{DPO}{}{}$):
\begin{equation}
\begin{aligned}
    \delta(R_C)*\delta(R_V)&=\delta(R_C\uplus R_V)+2\delta(R_C')\,,\; R_C':=\left(\inputtikz{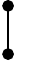}\hookleftarrow \inputtikz{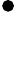}\hookrightarrow \inputtikz{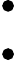};\ac{true}\right)\\
    \delta(R_V)*\delta(R_C)&=\delta(R_C\uplus R_V)\,.
\end{aligned}
\end{equation}
To provide some intuition: the first computation encodes the causal information that the two rules may either be composed along a trivial overlap, or rule $R_C$ may overlap on one of the vertices in the output of $R_V$; in the latter case, any vertex to which first $R_V$ and then $R_C$ applies must not have had any incident edges, i.e.\ in particular no edge violating the constraint of $R_C$, which is why the composite rule $R_C'$ does not feature any non-trivial constraint. In the other order of composition, the two vertices in the output of $R_C$ are linked by an edge, so $R_V$ cannot be applied to any of these two vertices (leaving just the trivial overlap contribution).
\end{example}

Just as the rule algebra construction encodes the compositional associativity property of rule compositions, the following \emph{representation} construction encodes in a certain sense the properties described by the concurrency theorem:
\begin{definition}
    Let $\bfC$ be an $\cM$-adhesive category satisfying Assumption~\ref{as:main}. Let $\hat{\bfC}$ be defined as the \emph{$\bR$-vector space} whose set of basis vectors is isomorphic to the set\footnote{We assume here that the isomorphism classes of objects of $\bfC$ form a \emph{set} (i.e.\ not a proper class).} of iso-classes of objects of $\bfC$ via a bijection $\ket{.}:\obj{\bfC}_{\cong}\rightarrow \mathsf{basis}(\hat{\bfC})$. Then the \emph{$\bT$-type canonical representation of the $\bT$-type rule algebra over $\bfC$}, denoted $\overline{\rho}^{\bT}_{\bfC}$, is defined as the morphism $\overline{\rho}^{\bT}_{\bfC}:\overline{\cR}_{\bfC}^{\bT}\rightarrow End_{\bR}(\hat{\bfC})$ specified via 
    \begin{equation}
    \forall R\in\LinEq{\bfC},X\in\obj{\bfC}_{\cong}:
    \quad \canRep{\bT}{\delta(R)}\ket{X}:=\sum_{m\in\MatchGT{\bT}{R}{X}}\ket{R_m(X)}\,.
    \end{equation}
\end{definition}

\begin{theorem}\label{thm:canrep}
    $\overline{\rho}^{\bT}_{\bfC}$ as defined above is an \emph{algebra homomorphism} (and thus in particular a well defined representation).
\end{theorem}
\begin{proof}
The proof is entirely analogous to the one for the case without application conditions~\cite{bp2018,nbSqPO2019} (cf.\ Appendix~\ref{app:proofCanrep}).
\end{proof}

\section{Stochastic mechanics formalism}\label{sec:stochMech}

Referring to~\cite{bdg2016,bp2019-ext,bdg2019} for further details and derivations, suffice it here to highlight the key role played by the algebraic concept of \emph{commutators} in stochastic mechanics. Let us first provide the constructions of continuous-time Markov chains (CTMCs) and observables in stochastic rewriting systems.
\begin{definition}
    Let $\bra{}:\hat{\bfC}\rightarrow \bR$ (referred to as \emph{dual projection vector}) be defined via its action on basis vectors of $\hat{\bfC}$ as $\braket{}{X}:=1_{\bR}$.
\end{definition}

\begin{theorem}\label{thm:CTMCs}
    Let $\bfC$ be a category satisfying the relevant variant of Assumption~\ref{as:main}, and let $\overline{\cR}_{\bfC}^{\bT}$ be the $\bT$-type rule algebra of linear rules with conditions over $\bfC$. Let $\rho\equiv \rho^{\bT}_{\bfC}$ denote the $\bT$-type canonical representation of $\cR_{\bfC}^{\bT}$. Then the following results hold:
\begin{enumerate}
    \item The basis elements of the space $\obs{\bfC}_{\bT}$ of \textbf{$\bT$-type observables}, i.e.\ the diagonal linear operators that arise as (linear combinations of) $\bT$-type canonical representations of rewriting rules with conditions, have the following structure ($ \hat{\cO}_{P,q}^{\ac{c}_P}$ in the DPO case, $\hat{\cO}_P^{\ac{c}_P}$ in the SqPO case):
\begin{equation}\label{eq:defObs}
\begin{aligned}
    \hat{\cO}_{P,q}^{\ac{c}_P}&:=\rho(\delta(P\xleftarrow{q}Q\xrightarrow{q}P;\ac{c}_P))\quad (P\in \obj{\bfC}_{\cong},q\in\cM,\ac{c}_P\in \cond{\bfC}_{\sim})\\
    \hat{\cO}_P^{\ac{c}_P}&:=\rho(\delta(P\xleftarrow{\cong}P\xrightarrow{\cong}P;\ac{c}_P))\quad (P\in \obj{\bfC}_{\cong},\ac{c}_P\in \cond{\bfC}_{\sim})\,.
\end{aligned}
\end{equation}
\item \textbf{DPO-type jump closure property:} for every linear rule with condition $R\equiv(O\hookleftarrow K\hookrightarrow I,\ac{c}_{I})\in \LinAc{\bfC}$, we find that
    \begin{equation}
        \bra{}\rho(\delta(R))=\bra{}\jcOp{\delta(R)}\,,
    \end{equation}
    where $\hat{\bO}:\overline{\cR}^{\text{DPO}}_{\bfC}\rightarrow End_{\bR}(\hat{\bfC})$ is the homomorphism defined via its action on basis elements $\delta(R)$ for $R=(O\hookleftarrow K\hookrightarrow I; \ac{c}_{I})\in\LinEq{\bfC}$ as
    \begin{equation}
        \jcOp{\delta(R)}:=\rho(\delta(I\hookleftarrow K\hookrightarrow I;\ac{c}_{I}))\in \obs{\bfC}\,.
    \end{equation}
    \item \textbf{SqPO-type jump closure property:} for every linear rule with condition $R\equiv(O\hookleftarrow K\hookrightarrow I,\ac{c}_{I})\in \LinAc{\bfC}$, we find that
    \begin{equation}
        \bra{}\rho(\delta(R))=\bra{}\jcOp{\delta(R)}\,,
    \end{equation}
    where\footnote{Since in applications we will always fix the type of rewriting to either DPO or SqPO, we will use the same symbol for the jump-closure operator in both cases.} $\hat{\bO}:\overline{\cR}^{\text{SqPO}}_{\bfC}\rightarrow End_{\bR}(\hat{\bfC})$ is the homomorphism defined via
    \begin{equation}
        \jcOp{\delta(R)}:=\rho(\delta(I\xleftarrow{\cong}I\xrightarrow{\cong}I;\ac{c}_{I}))\in \obs{\bfC}\,.
    \end{equation}
    \item \textbf{CTMCs via stochastic rewriting systems:} Let $\Prob{\bfC}$ be the space of \emph{(sub-)probability distributions over $\hat{\bfC}$} (i.e.\ $\ket{\Psi}=\sum_{X\in\obj{\bfC}_{\cong}}\psi_X\ket{X}$). Let $\cT$ be a collection of $N$ pairs of positive real-valued parameters $\kappa_j$ (referred to as \emph{base rates}) and linear rules $R_j$ with application conditions,
\begin{equation}
    \cT:=\{(\kappa_j,R_j)\}_{1\leq j\leq N}\qquad (\kappa_j\in \bR_{\geq 0}\,,\;R_j\equiv(r_j,\ac{c}_{I_j})\in \LinAc{\bfC})\,.
\end{equation} 
Then given an \emph{initial state} $\ket{\Psi_0}\in \Prob{\bfC}$, the $\bT$-type stochastic rewriting system based upon the transitions $\cT$ gives rise to the CTMC $(\cH,\ket{\Psi(0)})$ with time-dependent state $\ket{\Psi(t)}\in \Prob{\bfC}$ (for $t\geq0$) and evolution equation
\begin{equation}    
 \forall t\geq 0:\quad   \tfrac{d}{dt}\ket{\Psi(t)}=\cH\ket{\Psi(t)}\,,\quad \ket{\Psi(0)}=\ket{\Psi_0}\,.
\end{equation}
Here, the \emph{infinitesimal generator} $\cH$ of the CTMC is given by
\begin{equation}\label{eq:H}
    \cH=\hat{H}-\jcOp{\hat{H}}\,,\quad \hat{H}=\sum_{j=1}^N \kappa_j\,\rho(\delta(R_j))\,.
\end{equation}
\end{enumerate}
\end{theorem} 
\begin{proof}
	See Appendix~\ref{sec:CTMCproofsApp}.
\end{proof}

\begin{remark}
    The operation $\hat{\bO}$ featuring in the DPO- and SqPO-type jump-closure properties has a very intuitive interpretation: given a linear rule with condition $R\equiv(r,\ac{c}_{I})\in \LinAc{\bfC}$, the linear operator $\jcOp{\delta(R)}$ is an observable that evaluates on a basis vector $\ket{X}\in \hat{\bfC}$ as $\jcOp{\delta(R)}\ket{X}=(\text{\# of ways to apply $R$ to $X$})\cdot\ket{X}$.
\end{remark}

As for the concrete computational techniques offered by the stochastic mechanics formalism, one of the key advantages of this rule-algebraic framework is the possibility to reason about \emph{expectation values} (and higher moments) of pattern-count observables in a principled and universal manner. The precise formulation is given by the following generalization of results from~\cite{bdg2019} to the setting of DPO- and SqPO-type rewriting for rules with conditions:
\begin{theorem}
    Given a CTMC $(\ket{\Psi_0},\cH)$ with time-dependent state $\ket{\Psi(t)}$ (for $t\geq 0$), a set of observables $O_1,\dotsc O_n\in\obs{\bfC}$ and $n$ \emph{formal variables} $\lambda_1,\dotsc,\lambda_n$, define the \emph{exponential moment-generating function (EMGF)} $M(t;\vec{\lambda})$ as
    \begin{equation}
        M(t;\vec{\lambda}):=\bra{}e^{\vec{\lambda}\cdot\vec{O}}\ket{\Psi(t)}\,,\quad  \vec{\lambda}\cdot\vec{O}:=\sum_{j=1}^n \lambda_jO_j\,.
    \end{equation}
    Then $M(t;\vec{\lambda})$ satisfies the following \emph{formal evolution equation} (for $t\geq0$):
    \begin{equation}\label{eq:MEGFevo}
    \begin{aligned}
	\tfrac{d}{dt}M(t;\vec{\lambda})&=\sum_{q\geq 1}\tfrac{1}{q!}\bra{}\left(ad_{\vec{\lambda}\cdot\vec{O}}^{\circ q}(\hat{H})\right)e^{\vec{\lambda}\cdot\vec{O}}\ket{\Psi(t)}\,,\; M(0;\vec{\lambda})=\bra{}e^{\vec{\lambda}\cdot\vec{O}}\ket{\Psi_0}\,.
    \end{aligned}
    \end{equation}
\end{theorem}
\begin{proof}
	In full analogy to the case of rules without conditions~\cite{bdg2019}, the proof  follows from the BCH formula $e^{\lambda A}Be^{-\lambda A}=e^{ ad_{\lambda A}}(B)$ (for $A,B\in End_{\bR}(\hat{\bfC})$). Here, $ad_A^{\circ 0}(B):=B$, $ad_A(B):=AB-BA$  (also referred to as the \emph{commutator} $[A,B]$ of $A$ and $B$), and $ad_A^{\circ(q+1)}(B):=ad_A(ad_A^{\circ q}(B))$ for $q\geq 1$. Finally, the $q=0$ term in the above expression evaluates identically to $0$ due to $\bra{}\cH=0$.
\end{proof}

Combining this theorem with the notion of $\bT$-type jump-closure, one can in favorable cases express the EMGF evolution equation as a PDE on formal power series in $\lambda_1,\dotsc,\lambda_n$ and with $t$-dependent real-valued coefficients. Referring the interested readers to~\cite{bdg2019} for further details on this technique, let us provide here a simple non-trivial example of such a calculation.

\begin{example}\label{ex:ugModel}
    Let us consider a stochastic rewriting system over the category $\bfC=\mathbf{uGraph}$ of finite undirected multigraphs, with objects further constrained by the structure constraint $\ac{c}^S_{\mIO}:=\neg \exists(\mIO\hookrightarrow \inputtikz{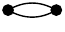})\in\cond{\mathbf{uGraph}}$ that prohibits multiedges. %
    Let us consider for type $\bT=SqPO$ the four rules with conditions $R_{E_{\pm}}$ (edge-creation/-deletion) and $R_{V_{\pm}}$ (vertex creation/deletion), defined as 
\begin{equation*}
\begin{aligned}
        R_{E_{+}}&:=\tfrac{1}{2}\delta\left(\inputtikz{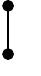}\hookleftarrow \inputtikz{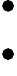}\hookrightarrow \inputtikz{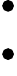};\neg \exists \left(\inputtikz{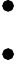}\hookrightarrow 
\inputtikz{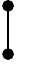}\right)\right)\,, &
R_{V_{+}}&:=\delta(\inputtikz{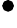}\hookleftarrow\mIO\hookrightarrow\mIO;\ac{true})\\
R_{E_{-}}&:=\tfrac{1}{2}\delta\left(\inputtikz{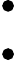}\hookleftarrow \inputtikz{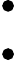}\hookrightarrow 
\inputtikz{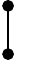};\ac{true}\right)\,,\quad &
R_{V_{-}}&:=\delta(\mIO\hookleftarrow\mIO\hookrightarrow
\inputtikz{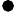};\ac{true})\,.
\end{aligned}
    \end{equation*}
Here, the prefactors $\tfrac{1}{2}$ for $R_{E_{\pm}}$ are chosen purely for convenience. Note that $R_{E_{+}}$ is the only rule requiring a non-trivial application condition, since linking two vertices with an edge might create a multiedge (precisely when the two vertices were already linked). Introducing base rates $\nu_{\pm},\varepsilon_{\pm}\in \bR_{>0}$ and letting $\hat{X}:=\rho(R_X)$,  we may assemble the infinitesimal generator $\cH$ of a CTMC as
\begin{equation}
\cH=\hat{H}+\jcOp{\hat{H}}\,,\;
\hat{H}:=\nu_{+}\hat{V}_{+}+\nu_{-}\hat{V}_{-}+\varepsilon_{+}\hat{E}_{+}+\varepsilon_{-}\hat{E}_{-}\,.
\end{equation}
One might now ask whether there is any interesting dynamical structure e.g.\ in the evolution of the moments of the observables that count the number of times each of the transitions of this system is applicable,
\begin{equation}
    O_{\bullet\vert \bullet}:=\jcOp{\delta(R_C)}\,,\;
    O_{\bullet\!-\! \bullet}:=\jcOp{\delta(R_D)}\,,\; O_{\bullet}:=\jcOp{\delta(R_{VD})}\,.
\end{equation}
The algebraic data necessary in order to formulate EMGF evolution equations are all 
\textbf{\emph{commutators}} of the observables with the contributions $\hat{X}:=\rho(\delta(R_{X}))$  to the ``off-diagonal part'' $\hat{H}$ of the infinitesimal generator $\cH$. We will present here for brevity just those commutators necessary in order to compute the evolution equations for the averages of the three observables:
\begin{equation}
\begin{aligned}
	[O_{\bullet},\hat{V}_{\pm}]&=\pm \hat{V}_{\pm}\,,\; & 	
	[O_{\bullet},\hat{E}_{\pm}]&=0 &&\\
	[O_{\bullet\vert \bullet},\hat{V}_{+}]&= \hat{A}\,, &
	[O_{\bullet\vert \bullet},\hat{V}_{-}]&= -\hat{B}\,,\; &
	[O_{\bullet\vert \bullet},\hat{E}_{\pm}]&= \mp \hat{E}_{\pm}\\
	[O_{\bullet\!-\! \bullet},\hat{V}_{+}]&=0\,, &
	[O_{\bullet\!-\! \bullet},\hat{V}_{-}]&= -\hat{C}\,,\; &
	[O_{\bullet\!-\! \bullet},\hat{E}_{\pm}]&= \pm \hat{E}_{\pm}
\end{aligned}
\end{equation}
As typical in these types of commutator computations, we find a number of contributions (here $\hat{A}$, $\hat{B}$ and $\hat{C}$) that were not either observables or based upon rules of the SRS:
\begin{equation*}
\begin{aligned}
\hat{A}&:=\rho\left(\delta\left(\inputtikz{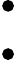}\hookleftarrow \inputtikz{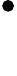}\hookrightarrow \inputtikz{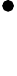};\ac{true}\right)\right),\;
\hat{B}:=\rho\left(\delta\left(\inputtikz{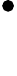}\hookleftarrow \inputtikz{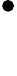}\hookrightarrow \inputtikz{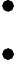};\neg \exists \left(\inputtikz{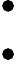}\hookrightarrow \inputtikz{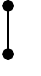}\right)\right)\right)\\
\hat{C}&:=\rho\left(\delta\left(\inputtikz{RC1}\hookleftarrow \inputtikz{RC2}\hookrightarrow \inputtikz{RC3};\ac{true}\right)\right),\; 
\jcOp{\hat{A}}=O_{\bullet}\,,\; \jcOp{\hat{B}}=2O_{\bullet\vert\bullet}\,,\;
\jcOp{\hat{C}}=2O_{\bullet\!-\!\bullet}
\end{aligned}
\end{equation*}
Picking for simplicity as an initial state $\ket{\Psi(0)}=\ket{\mIO}$ just the empty graph, and invoking the SqPO-type jump-closure property (cf.\ Theorem~\ref{thm:CTMCs}) repeatedly in order to evaluate $\langle [O_P,\hat{H}]\rangle(t)=\langle \jcOp{[O_P,\hat{H}]}\rangle(t)$, the moment EGF evolution equation~\eqref{eq:MEGFevo} specializes to the following ``Ehrenfest-like''~\cite{bdg2019} ODE system:
\begin{equation*}
\begin{aligned}
\tfrac{d}{dt}\langle O_{\bullet}\rangle(t)&=\langle [O_{\bullet},H]\rangle(t)=\nu_{+}-\nu_{-}\langle O_{\bullet}\rangle(t)\\
\tfrac{d}{dt}\langle O_{\bullet\vert\bullet}\rangle(t)&=\langle [O_{\bullet\vert\bullet},H]\rangle(t)
=\nu_{+}\langle O_{\bullet}\rangle(t)
-(2\nu_{-}+\varepsilon_{+})\langle O_{\bullet\vert\bullet}\rangle(t)
+\varepsilon_{-}\langle O_{\bullet\!-\!\bullet}\rangle(t)\\
\tfrac{d}{dt}\langle O_{\bullet\!-\!\bullet}\rangle(t)&=\langle [\langle O_{\bullet\!-\!\bullet}\rangle(t),H]\rangle(t)
=\varepsilon_{+}\langle O_{\bullet\vert\bullet}\rangle(t)
-(2\nu_{-}+\varepsilon_{-})\langle O_{\bullet\!-\!\bullet}\rangle(t)\\
\langle O_{\bullet}\rangle(0)&=\langle O_{\bullet\vert\bullet}\rangle(t)=\langle O_{\bullet\!-\!\bullet}\rangle(t)=0\,.
\end{aligned}
\end{equation*}
This ODE system may be solved exactly (see Appendix~\ref{app:se}). We depict in Figure~\ref{fig:evoEx} two exemplary evolutions of the three average pattern counts for different choices of parameters. Since due to SqPO-semantics the vertex deletion and creation transitions are entirely independent of the edge creation and deletion transitions, the vertex counts stabilize on a Poisson distribution of parameter $\nu_{+}/\nu_{-}$ (where we only present the average vertex count value here). As for the non-linked vertex pair and edge patter counts, the precise average values are sensitive to the parameter choices (i.e.\ whether or not vertices tend to be linked by an edge or not may be freely tuned in this model via adjusting the parameters).
\end{example}

\begin{figure}[t]\label{fig:evoEx}
 \centering
    \subfigure[\label{fig:a}Vertices tend to be linked.]{\includegraphics[width=0.45\textwidth]{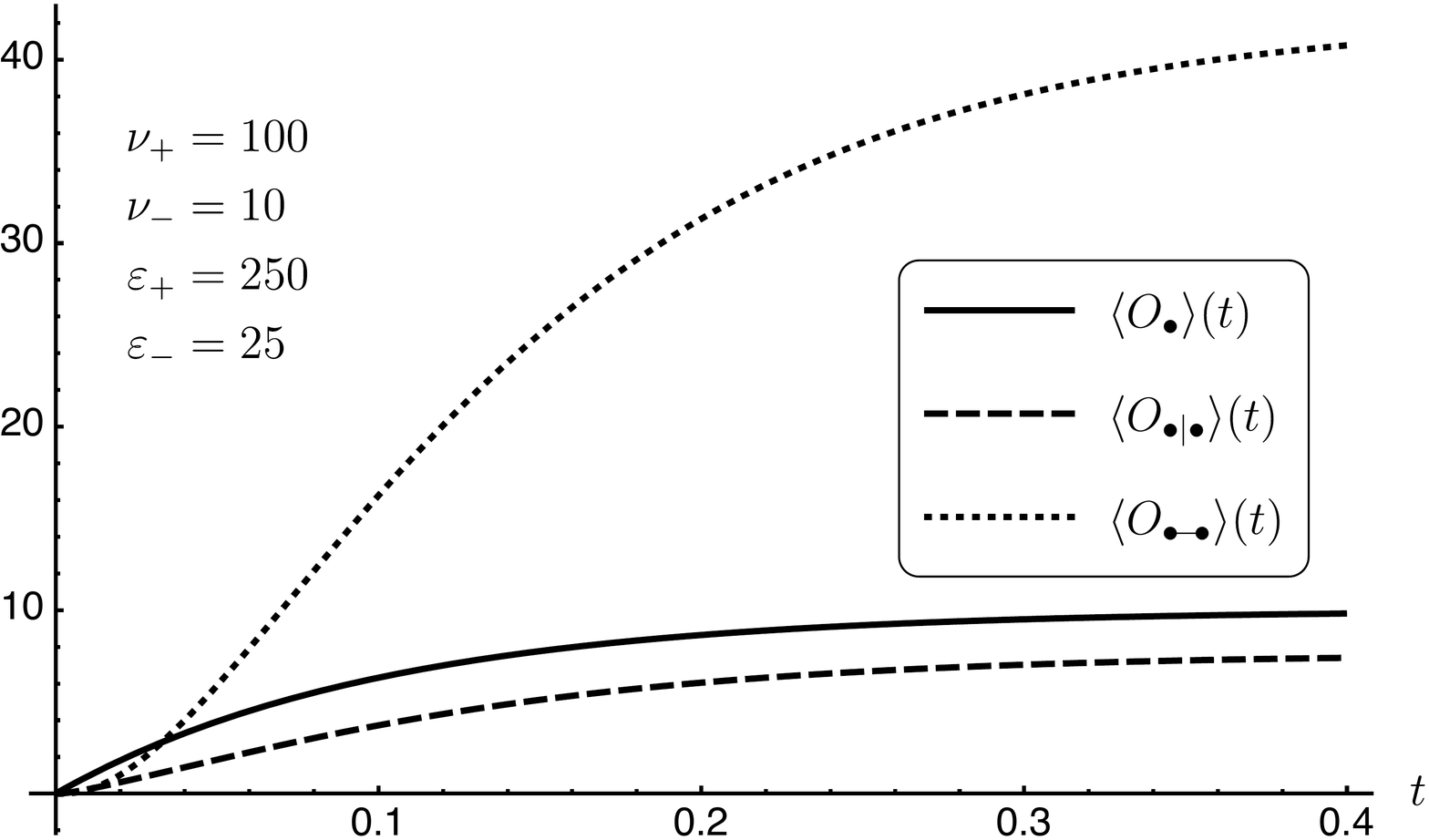}}
     \subfigure[\label{fig:b}Vertices tend to be unlinked.]{\includegraphics[width=0.45\textwidth]{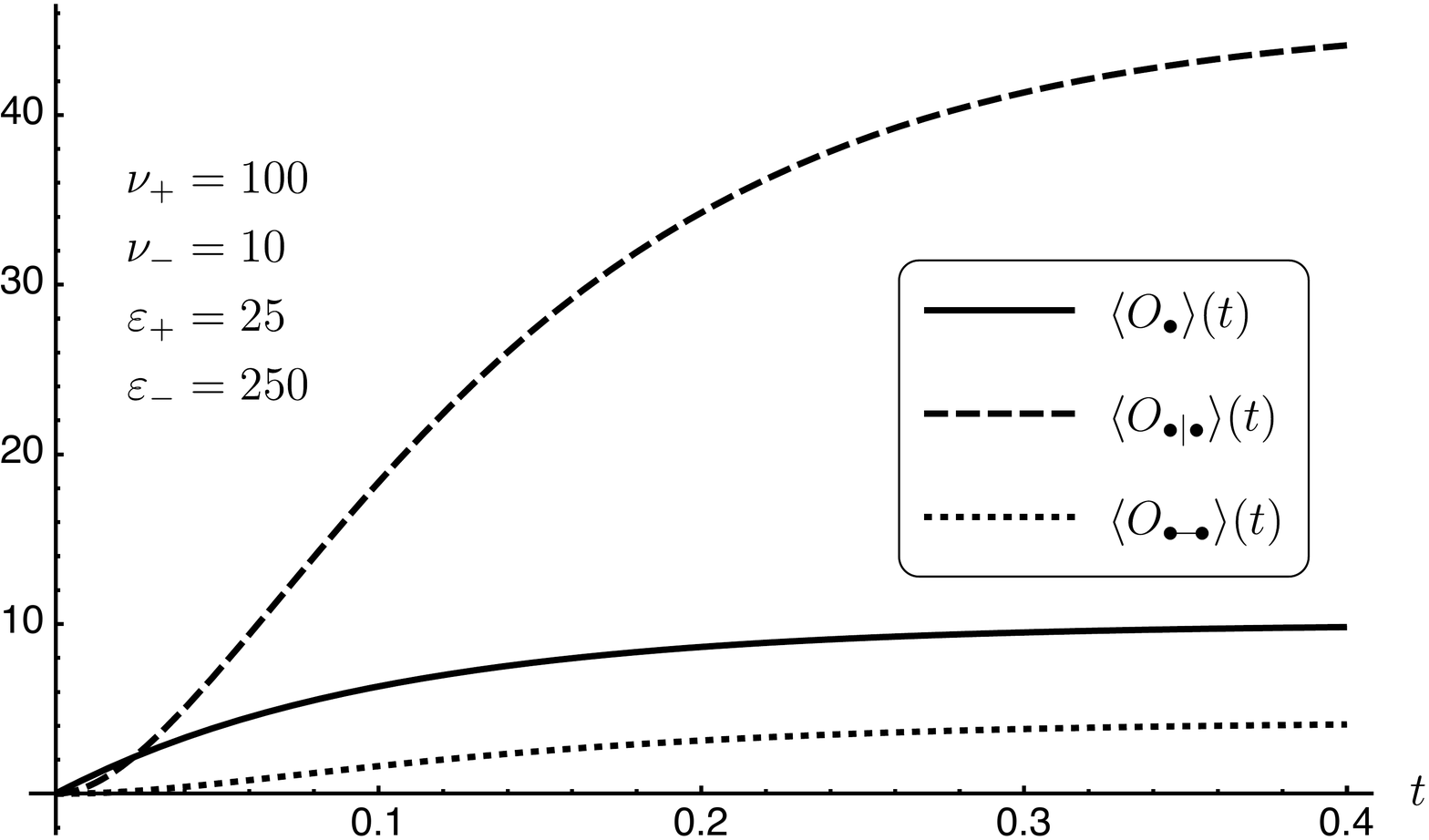}}
\caption{Time-evolutions of pattern count observables for different parameter choices.}
\end{figure}

While the example presented was chosen mainly to illustrate the computational techniques, it highlights the typical feature of the emergence of contributions in the relevant (nested) commutator calculations that may not have been included in the non-diagonal part $\hat{H}$ of the infinitesimal generator of the CTMC. We refer the interested readers to~\cite{bdg2019} for an extended discussion of this phenomenon, and for computation strategies for higher-order moment evolution equations.

\section{Application scenario 1: biochemistry with \KAP{}}\label{sec:bcgr}

The \href{https://kappalanguage.org}{\KAP{} platform}~\cite{danos2004computational,danos2004formal} for rule-based modeling of biochemical reaction systems is based upon the notion of so-called \emph{site-graphs} that abstract proteins and other complex macro-molecules into \emph{agents} (with \emph{sites} representing interaction capacities of the molecules). This open source platform offers a variety of high-performance \emph{simulation algorithms} (for CTMCs based upon \KAP{} rewriting rules) as well as several variants of static analysis tools to analyze and verify biochemical models~\cite{Boutillier:2018aa}. In view of the present paper, it is interesting to note that since the start of the \KAP{} development, the simulation-based algorithms have been augmented by \emph{differential semantics} modules aimed at deriving ODE systems for the evolution of pattern-count observable average values~\cite{danos2008rule,Danos_2010,danos2015moment,Harmer_2010}. In this section, we will experiment with a (re-)encoding of \KAP{} in terms of typed undirected graphs with certain structural constraints that permits to express such moment statistics ODEs via our general rule-algebraic stochastic mechanics formalism. We will then provide an illustrative exemplary computation of ODEs in order to point out certain intrinsic intricacies (notably non-closure properties) typical of such calculations. One of the key theoretical features of \KAP{} is its foundation upon the notion of \emph{rigidity}~\cite{Danos_2014}. In practice, the construction involves an \emph{ambient category} $\bA$ (which possesses suitable adhesivity properties), a \emph{pattern category} $\bP$ (obtained from $\bA$ via certain \emph{negative constraints}) and finally a \emph{state category} $\bS$ (obtained from $\bP$ via additional \emph{positive constraints}). We will now present one possible realization of \KAP{} based upon the $\cM$-adhesive category of typed undirected multigraphs:
\begin{definition}
For a \KAP{} model $K$, let $\bA=\mathbf{uGraph}/T_K$ be the category of finite undirected multigraphs typed over $T_K$, where $T_K$ distinguishes \emph{agent vertex types}, \emph{site vertex types} and three forms of \emph{edge types}: agent-site, site-site and loops on sites. For each agent type vertex $X\in\{A,B,\ldots\}$, the type graph contains the site type vertices $x_1:X,\dotsc,x_{n_X}:X$ (incident to the $X$-type vertex via an edge, and where $n_X<\infty$). $T_K$ also contains link type edges between sites that encode which sites can be linked, and loops on site type vertices that represent \emph{dynamic attributes}, such as the phosphorylation state of a site. Indicating the three different edge types by wavy lines (agent-site), solid lines (site-site) and dotted lines (property loops), the agent vertices with filled circles $\inputtikz{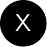}$ and the site vertices by open circles $\inputtikz{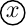}$, and using the placeholder $\bullet$ for a vertex and a dashed line for an edge of any type, we may introduce the \emph{negative constraints} defining the \emph{pattern category} $\bP_K$ as $\ac{c}_{\cN_K}:=\land_{N\in \cN_K}\neg\exists(\mIO\hookrightarrow N)$, with the set $\cN_K$ of \emph{``forbidden subgraphs''} defined as
\begin{equation}
	\cN_K:=\left\{
	\inputtikz{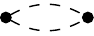}\right\} \cup\bigcup\limits_x\left\{
	\inputtikz{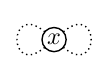}\right\}\cup \bigcup\limits_{X,x}\left\{
	\inputtikz{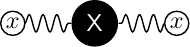}\,,\;\inputtikz{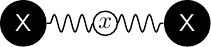}
	\right\}\,.
\end{equation}
Finally, the \emph{state category} $\bS_K$ is obtained from $\bP_K$ via imposing a positive constraint $\ac{c}_{\cP_K}$ that ensures that each agent $X$ is linked to exactly one of each of its site vertices $x:X$, and if a site $x:X$ can carry a property or alternative variants thereof, it also carries a loop that signifies one of these properties (see the example below for further details). Moreover, a given site $x:X$ must be linked to an agent $X$ (i.e.\ cannot occur in isolation).
\end{definition}
\begin{example}
	Consider a simple \KAP{} model with a type graph as below left that introduces two agent types $\mathsf{K}$ (for ``kinase'') and $\mathsf{P}$ (for ``protein''), where $\mathsf{K}$ has a site $k:\mathsf{K}$, and where $\mathsf{P}$ has sites $p_t,p_l,p_b:\mathsf{P}$. Moreover, the sites $p_t$ and $p_b$ can carry properties $\mathsf{u}$ (``unphosphorylated'') and $\mathsf{p}$ (``phosphorylated''), depicted as dotted loops in the type graph. Sites $k:\mathsf{K}$ and $p_l:\mathsf{P}$ can bind (as indicated by the solid line in the type graph).
\begin{equation*}
\begin{array}{cc|ccc|c}
\inputtikz{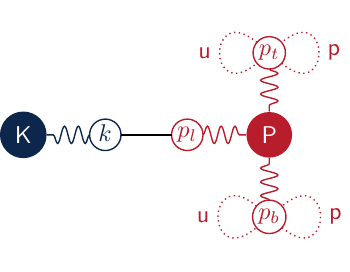} &\hphantom{x} &\hphantom{x} &
\begin{array}{rcl}
\mIO
	& \xrightleftharpoons[\;k_{-}\;]{k_{+}} &
\inputtikz{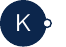}\\
\inputtikz{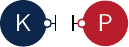}
	& \xrightleftharpoons[\;l_{-}\;]{l_{+}} &
\inputtikz{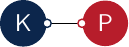}\\
 \inputtikz{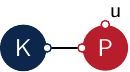} &\xrightleftharpoons[\;t_{-}\;]{t_{+}} &
\inputtikz{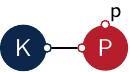}\\
 \inputtikz{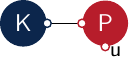} &\xrightleftharpoons[\;b_{-}\;]{b_{+}} &
\inputtikz{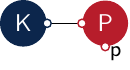}
\end{array} &\hphantom{x}&\hphantom{x}
\begin{array}{rcl}
\inputtikz{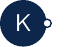}\text{$\!\!$} &\xleftharpoonup{\,r_{obs_K}\,}&
 \inputtikz{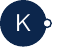}
\\
\\
\inputtikz{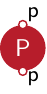} \text{$\,$} &\xleftharpoonup{\,r_{obs_P}\,}&
 \inputtikz{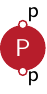}
\end{array}
\end{array} 
\end{equation*}	
As a prototypical example of a \KAP{} stochastic rewriting system, consider a system based upon the rewriting rules $k_{\pm}$, $l_{\pm}$, $t_{\pm}$ and $b_{\pm}$. Here, for the rule $l_{+}$, we have indicated that it must be equipped with an application condition that ensures that the site of the $\mathsf{K}$-type agent and the left site of the $\mathsf{P}$-type agent must be \emph{free} before binding. As common practice also in the standard \KAP{} theory, we otherwise leave in the graphical depictions those application conditions necessary to ensure consistent matches implicit as much as possible. Consider then for a concrete computational example the time-evolution of the average count of the pattern described in the identity rule $r_{obs_P}$. As typical in \KAP{} rule specifications $r_{obs_P}$ as well as several of the other rules depicted only explicitly involve \emph{patterns}, but not necessarily \emph{states}, since e.g.\ in $r_{obs_P}$ the left site of the $\mathsf{P}$-type agent is not mentioned. In complete analogy to the computation presented in Example~\ref{ex:ugModel}, let us first compute the commutators of the observable  $O_{\mathsf{K}}=\rho(\delta(r_{obs_K};\ac{c}_{obs_K}))$ with the operators $\hat{X}:=\rho(\delta(r_X;\ac{c}_X))$:
\begin{equation}
\begin{aligned}
	[O_{\mathsf{K}},\hat{K}_{\pm}]&=\pm \hat{K}_{\pm}\,,\; [O_{\mathsf{K}},\hat{L}_{\pm}]=[O_{\mathsf{K}},\hat{T}_{\pm}]=[O_{\mathsf{K}},\hat{B}_{\pm}]=0
\end{aligned}
\end{equation}
However, letting $O^{(\mathsf{x},\mathsf{y})}_{P}$, $O^{(\mathsf{x},\mathsf{y})}_{link}$ and $O^{(\mathsf{x},\mathsf{y})}_{free}$ denote the observables for the patterns
\[
\omega_P^{(\mathsf{x},\mathsf{y})}:=\inputtikz{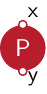}\,,\quad \omega_{link}^{(\mathsf{x},\mathsf{y})}:=
\inputtikz{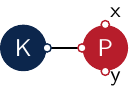}\,,\quad \omega_{free}^{(\mathsf{x},\mathsf{y})}:=
\inputtikz{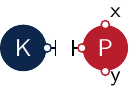}
\]
one may easily demonstrate that even a comparatively simple observable such as $O^{(\mathsf{p},\mathsf{p})}_{P}$ already leads to an infinite cascade of  contributions to the ODEs for the averages of pattern counts. As typical in these sorts of computations, the discovery of a new pattern observable via applying SqPO-type jump-closure (Theorem~\ref{thm:CTMCs}) to the commutator contributions to $\tfrac{d}{dt}\langle O^{(\mathsf{p},\mathsf{p})}_{P}\rangle(t)$ leads to the discovery of new pattern observables yet again, such as in
\[
	[O_{\mathsf{P}},\hat{T}_{+}]=\hat{T}_{+}^{(\mathsf{p})}\,,\;
	\jcOp{\hat{T}_{+}^{(\mathsf{p})}}=O^{(\mathsf{u},\mathsf{p})}_{link}\,,\;
	[O^{(\mathsf{u},\mathsf{p})}_{link},\hat{L}_{+}]=\hat{L}^{(\mathsf{u},\mathsf{p})}\,,\;
	\jcOp{\hat{L}^{(\mathsf{u},\mathsf{p})}}=O^{(\mathsf{u},\mathsf{p})}_{free}\,.
\]
In particular the last observable $O^{(\mathsf{u},\mathsf{p})}_{free}$ is found to lead to an infinite tower of other observables (i.e.\ ``ODE non-closure''), starting from
\begin{equation*}
	[O^{(\mathsf{u},\mathsf{p})}_{free},\hat{L}_{+}]=-\hat{L}^{(\mathsf{u},\mathsf{p})}
	-\left(\inputtikz{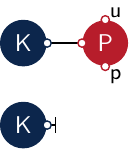}\leftharpoonup \inputtikz{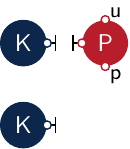}
	\right)-\left(\inputtikz{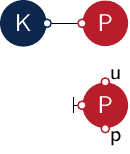}\leftharpoonup \inputtikz{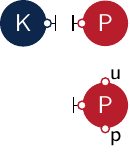}
	\right)\,.
\end{equation*}
\end{example}

This exemplary and preliminary analysis reveals that while the rule-algebraic CTMC implementation is in principle applicable to the formulation and analysis  \KAP{}  systems, further algorithmic and theoretical developments will be necessary (including possibly ideas of \emph{fragments} and \emph{refinements} as in~~\cite{danos2008rule,Danos_2010,Harmer_2010}) in order to obtain a computationally useful alternative rewriting-theoretic implementation of \KAP{}.

\section{Application scenario 2: organic chemistry with \MOD{}}\label{sec:ocgr}

The \href{https://cheminf.imada.sdu.dk/mod/}{\MOD{} platform}~\cite{Andersen_2016} for organo-chemical reaction systems is a prominent example of a DPO-type rewriting theory of high relevance to the life sciences. From a theoretical perspective, \MOD{} has been designed~\cite{andersen2018rule} as a rewriting system over so-called \emph{chemical graphs}, a certain typed and undirected variant of the category $\mathbf{PLG}$ of partially labelled directed graphs. While the latter category had been introduced in~\cite{Habel_2012} as a key example of an $\cM$-$\cN$-adhesive category, with the motivation of permitting label-changes in rewriting rules, it was also demonstrated in loc cit.\ that $\mathbf{PLG}$ is \emph{not} $\cM$-adhesive. Since moreover no concrete construction of a tentative variant $\mathbf{uPLG}$ of $\mathbf{PLG}$ for undirected graphs, let alone results on the possible adhesivity properties of such a category are known in the literature, we propose here an alternative and equivalent encoding of chemical graphs. We mirror the constructions of~\cite{Andersen_2016,andersen2018rule} in that chemical graphs will be a certain typed variant of undirected graphs, with \emph{vertex types} representing \emph{atom types}, \emph{edge types} ranging over the types $\{\mathtt{-},\mathtt{=},\mathtt{\#},\mathtt{:}\}$ representing \emph{single, double,triple and aromatic bonds}, respectively, and with the graphs being required to not contain multiedges. Inspired by the \KAP{} constructions in the previous section, we opt to represent \emph{properties} (such as e.g.\ \emph{charges} on atoms) as \emph{typed loop edges} on vertices representing atoms, whence the change of a property (which was the main motivation in~\cite{andersen2018rule} for utilizing a variant of $\mathbf{PLG}$) may be encoded in a rewriting rule simply via deletion/creation of property-encoding loops. Unfortunately, while the heuristics presented thus far would suggest that chemical graphs in the alternative categorical setting should be just simple typed undirected graphs, the full specification of chemical graphs would also have to include additional, empirical information from the chemistry literature. Concretely, atoms such as e.g.\ carbon only support a limited variety of bond types and configurations of incident bonds (referred to as \emph{valencies}), with additional complications such as poly-valencies possible for some types of atoms as illustrated by the following example.

\begin{example} The \emph{Meisenheimer-2-3-rearrangement} reaction~\cite{March2001} (cf.\ also~\cite{Andersen_2017}) constitutes an example\footnote{This example reaction was typeset directly via \MOD{} (cf.\ Appendix~\ref{app:ME}).} of a reaction where \emph{polyvalence} is encountered:
\begin{equation}
\includegraphics[width=0.8\textwidth]{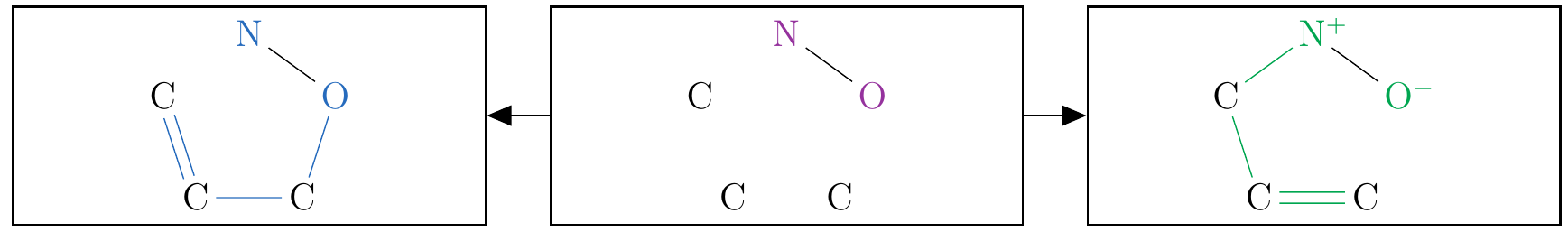}
\end{equation}
Upon matching this rule into a chemically valid mixture, the $N$ atom on the input of the rule will have valence $5$, while on the output it will have valence $3$. This type of information is evidently in no way contained in the chemical graphs alone, and must therefore be encoded in terms of suitable additional typing on the graphs and application conditions.
\end{example}

While thus at present no encoding of chemical graphs into a categorical framework with suitable adhesivity properties is available, we posit that it would be highly fruitful in light of the stochastic mechanics framework presented in this paper to develop such an encoding (joint work in progress with J.L.\ Andersen, W.\ Fontana and D.\ Merkle).

\section{Conclusion and outlook}

Rewriting theories of DPO- and SqPO-type for rules with conditions over $\cM$-adhesive categories are poised to provide a rich theoretical and algorithmic framework for modeling stochastic dynamical systems in the life sciences. The main result of the present paper consists in the introduction of a \emph{rule algebra framework} that extends the pre-existing constructions~\cite{bp2018,nbSqPO2019,bdg2016} precisely via incorporating the notion of conditions. The sophisticated \KAP{}~\cite{Boutillier:2018aa} and \MOD{}~\cite{Andersen_2016} bio-/organo-chemistry platforms and related developments have posed one of the main motivations for this work. For both of these platforms, we present a first analysis and stepping stones towards bridging category-theoretical rewriting theories and stochastic mechanics computations. Especially for the organo-chemistry setting, our work motivates the development of a full encoding  of (at least a reasonable fragment of) organic chemistry in terms of \emph{chemical graphs} and rewriting rules thereof, which to date is still unavailable. This encoding will be beneficial also in the development of tracelet-based techniques~\cite{behr2019tracelets}, and is current work in progress.

An intriguing perspective for future developments in categorical rewriting theory consists in developing a robust and versatile methodology for the analysis of ODE systems of pattern-counting observables in stochastic rewriting systems. %
While the results of this paper permit to formulate dynamical evolution equations for arbitrary higher moments of such observables, in general cases (as illustrated in Section~\ref{sec:bcgr}) the non-closure of the resulting ODE systems remains a fundamental technical challenge. In the \KAP{} literature, sophisticated conceptual and algorithmic approaches to tackle this problem have been developed such as refinements~\cite{danos2008rule,Danos_2014}, model reduction techniques~\cite{Danos_2010} and stochastic fragments~\cite{ferethal00975861} (see also~\cite{bdg2019} for an extended discussion). We envision that a detailed understanding of these approaches from within the setting of   categorical rewriting and of rule algebra theory could provide a very fruitful enrichment of the methodology of rewriting theory.



\clearpage
\appendix

\section{Background material on adhesive categories and rewriting with conditions}\label{sec:appendixA}

As a reference for notational conventions and in order to recall some of the standard definitions necessary in the main text, we collect here some of the materials contained in our recent paper~\cite{behrRaSiR} for the readers' convenience.

\subsection{$\cM$-adhesive categories}\label{sec:MACapp}

\begin{definition}
An $\cM$-adhesive category~\cite{ehrig2010categorical} $(\bfC,\cM)$ is a category $\bfC$ together with a class of monomorphisms $\cM$ that satisfies the following properties:
\begin{enumerate}
\item $\bfC$ has pushouts and pullbacks along\footnotemark $\cM$-morphisms.
\item The class $\cM$ contains all isomorphisms and is stable under pushout, pullback and composition.
\item Pushouts along $\cM$-morphisms are $\cM$-van Kampen squares.
\end{enumerate}
\begin{minipage}[t]{0.6\linewidth}
The latter property entails that in a commutative diagram such as the one on the right where the bottom square is a pushout along an $\cM$-morphism, where the back and right faces pullbacks and where all vertical morphisms are in $\cM$, the bottom square is $\cM$-van Kampen if the following property holds: the top square is a pushout if and only if the front and left squares are pullbacks.
\end{minipage}%
\begin{minipage}[t]{0.4\linewidth}
\null\hfill\\[-\dimexpr\baselineskip+1.2em\relax]
\centering
\text{$\;$}\includegraphics{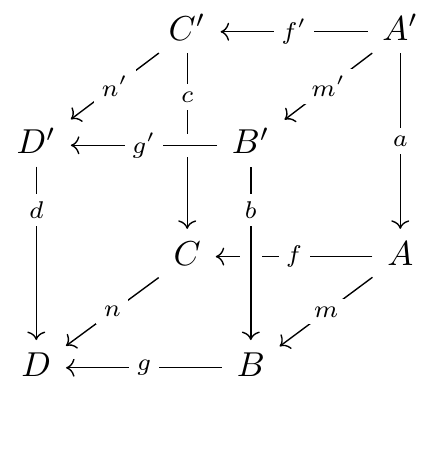}
\end{minipage}
\end{definition}
\footnotetext{Here, ``along'' entails that at least one of the two morphisms involved in the relevant (co-)span is in $\cM$.}%
Throughout the following definitions, let $(\bfC,\cM)$ be an $\cM$-adhesive category.
\begin{definition}
    $(\bfC,\cM)$ is said to be \emph{finitary}~\cite{GABRIEL_2014} if every object has only finitely many $\cM$-subobjects up to isomorphism.
\end{definition}
\begin{definition}
    $(\bfC,\cM)$ possesses an $\cM$-initial object $\mIO$~\cite{GABRIEL_2014} if for all objects $X\in \obj{\bfC}$ there exists a unique $\cM$-morphism $\iota_X:\mIO\hookrightarrow X$.
\end{definition}
\begin{definition}
    $(\bfC,\cM)$ possesses an \emph{epi-$\cM$-factorization}~\cite{habel2009correctness} if every morphism $f\in \mor{\bfC}$ factorizes as $f=m\circ e$ with $m\in \cM$ and with $e\in \epi{\bfC}$ an epimorphism, and such that this factorization is unique up to isomorphism.
\end{definition}
\begin{definition}
    $(\bfC,\cM)$ has \emph{$\cM$-effective unions} if for every cospan $(B\hookrightarrow D\hookleftarrow C)$ of $\cM$-morphisms that is the pushout of a span $(B\hookleftarrow A\hookrightarrow C)$, the following property holds: for every cospan $(B\hookrightarrow E\hookleftarrow C)$ whose pullback is given by $(B\hookleftarrow A\hookrightarrow C)$, the morphism $D\rightarrow E$ that exists by universal property of the pushout is in $\cM$.
\end{definition}

We next recall the notion of final pullback complements that is an important technical ingredient of the theory of SqPO-rewriting.

\begin{definition}
    Let $(b,a)$ be a composable pair of morphisms in a category $\bfC$. Then a pair of morphisms $(c,d)$ is called a \emph{final pullback complement (FPC)}~\cite{Corradini_2006} if $(a,d)$ is the pullback of $(b,c)$, and if for every $(a\circ p,q)$ that is the pullback of $(b,r)$, there exists a morphism $s$ such that $r=c\circ s$ that is unique up to isomorphism.
    \[
        \includegraphics{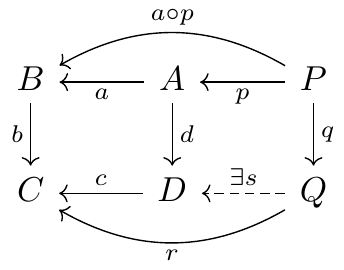}
    \]
\end{definition}

\begin{definition}
    The class of monomorphisms $\cM$ of $(\bfC,\cM)$ is said to be \emph{stable under FPCs}~\cite{behrRaSiR} if for every pair $(b,a)$ of composable $\cM$-morphisms the FPC $(c,d)$ (if it exists) is a pair of $\cM$-morphisms.
\end{definition}

\subsection{Conditions}\label{sec:condApp}

\begin{definition}
    \emph{Conditions}\cite{habel2009correctness,ehrig2014mathcal} in an $\cM$-adhesive category $(\bfC,\cM)$ satisfying Assumption~\ref{as:main} are recursively defined for every object $X\in \obj{\bfC}$ as follows:
    \begin{enumerate}
        \item $\ac{true}_X$ is a condition.
        \item Given $(f:X\hookrightarrow Y)\in \cM$ and a condition $\ac{c}_Y$, $\exists(f,\ac{c}_Y)$ is a condition.
        \item If $\ac{c}_X$ is a condition, so is $\neg \ac{c}_X$.
        \item If $\ac{c}_X^{(1)},\ac{c}_X^{(2)}$ are conditions, so is $\ac{c}_X^{(1)}\land \ac{c}_X^{(2)}$.
    \end{enumerate}
    The \emph{satisfaction} of a condition $\ac{c}_X$ by a $\cM$-morphism $(h:X\hookrightarrow Z)\in \cM$, denoted $h\vDash \ac{c}_X$, is recursively defined (with notations as above) as follows:
    \begin{minipage}[t]{0.6\textwidth}
    \begin{enumerate}
        \item $h\vDash\ac{true}_X$.
        \item $h\vDash \exists(f,\ac{c}_Y)$ iff there exists an $\cM$-morphism $(g:Y\hookrightarrow Z)\in \cM$ such that $h=g\circ f$ and $g\vDash Y$.
        \item $h\vDash \neg\ac{c}_X$ iff $h\not{\vDash} \ac{c}_X$.
        \item $h\vDash (\ac{c}_X^{(1)}\land \ac{c}_X^{(2)})$ iff $h\vDash \ac{c}_X^{(1)}$ and $h\vDash \ac{c}_X^{(2)}$.
    \end{enumerate}
    \end{minipage}
    \begin{minipage}[t]{0.4\textwidth}
    \null\hfill\\[-\dimexpr\baselineskip+0em\relax]
    \centering
    \includegraphics{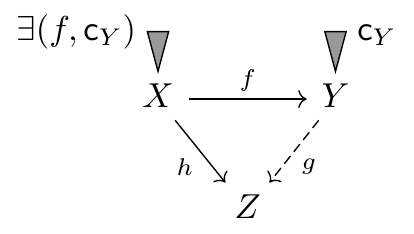}
    \end{minipage}
    \medskip

    Two conditions $\ac{c}_X$ and $\ac{c}_X'$ are \emph{equivalent}, denoted $\ac{c}_X\equiv \ac{c}_X'$, iff for every $\cM$-morphism $(h:X\hookrightarrow Z)\in \cM$, $h\vDash \ac{c}_X$ if and only if $h\vDash\ac{c}_X'$.
    \medskip

    Finally, a condition $\ac{c}_{\mIO}$ over the $\cM$-initial object $\mIO$ is called a \emph{constraint}, and we define for every object $Z\in \obj{\bfC}$
    \begin{equation}
        Z\vDash \ac{c}_{\mIO} \quad :\Leftrightarrow \quad (\mIO\hookrightarrow Z)\vDash \ac{c}_{\mIO}\,.
    \end{equation} 
\end{definition}

\begin{theorem}[\cite{habel2009correctness}; cf.\ also~\cite{behrRaSiR}]
    In an $\cM$-adhesive category satisfying Assumption~\ref{as:main}, there exists a \emph{shift operation}, denoted $\Shift$, such that for all conditions $\ac{c}_X$ and for all $\cM$-morphisms $(f:X\hookrightarrow Y)\in \cM$, $(g:Y\hookrightarrow Z)\in \cM$ and $(h:X\hookrightarrow Z)\in \cM$ with $h=g\circ f$, the following property holds:
    \begin{equation}
        h\vDash \ac{c}_X\quad \Leftrightarrow \quad g\vDash \Shift(f,\ac{c}_X)
    \end{equation}
\end{theorem}
We refer the interested readers to~\cite{behrRaSiR} for further details on the concrete implementation of the shift construction.

\subsection{Associativity and concurrency theorems}\label{app:ACthms}

In the statements of the following two theorems, we always imply choosing concrete representatives of the relevant equivalence classes of rules with conditions in order to list the sets of admissible matches.

\begin{theorem}[Associativity of rule compositions~\cite{bp2018,nbSqPO2019,behrRaSiR}]\label{thm:assocR}
    Let $\bfC$ be a category satisfying Assumption~\ref{as:main}. let $R_1,R_2,R_3\in \LinAc{\bfC}$ be linear rules with conditions, and let $\bT\in\{DPO,SqPO\}$. Then there exists a bijection $\varphi:A\xrightarrow{\cong} B$ of sets of pairs of $\bT$-admissible matches $A$ and $B$, defined as
    \begin{equation}
        \begin{aligned}
            A&:=\{(\mu_{21},\mu_{3(21)})\mid \mu_{21}\in 
            \MatchGT{\bT}{R_2}{R_1}\,,\; \mu_{3(21)}\in \MatchGT{\bT}{R_3}{R_{21}}\}\\
            B&:=\{(\mu_{32},\mu_{(32)1})\mid 
            \mu_{32}\in\MatchGT{\bT}{R_3}{R_2}\,,\;
             \mu_{(32)1}\in \MatchGT{\bT}{R_{32}}{R_1}\}\,,
        \end{aligned}
        \end{equation}
        where $R_{21}=\compGT{\bT}{R_2}{\mu_{21}}{R_1}$ and $R_{32}=\compGT{\bT}{R_3}{\mu_{32}}{R_2}$, such that for each corresponding pair $(\mu_{21},\mu_{3(21)})\in A$ and %
        $\varphi(\mu_{21},\mu_{3(21)})=(\mu_{32}',\mu_{(32)1}')\in B$, 
        \begin{equation}
            \compGT{\bT}{R_3}{\mu_{3(21)}}{\left(\compGT{\bT}{R_2}{\mu_{21}}{R_1}\right)}\cong
            \compGT{\bT}{\left(\compGT{\bT}{R_3}{\mu_{32}'}{R_2}\right)}{\mu_{(32)1}'}{R_1}\,.
        \end{equation}
    In this particular sense, the composition operations $\compGT{\bT}{.}{.}{.}$ are \textbf{associative}.
\end{theorem}

\begin{theorem}[Concurrency theorem~\cite{bp2018,nbSqPO2019,behrRaSiR}]\label{thm:concur}
    Let $\bfC$ be a category satisfying Assumption~\ref{as:main}, and let $\bT\in\{DPO,SqPO\}$. Then there exists a bijection $\varphi:A\xrightarrow{\cong}B$ on pairs of $\bT$-admissible matches between the sets $A$ and $B$,
    \begin{equation}
        \begin{aligned}
            A&=\{(m_2,m_1)\mid m_1\in \MatchGT{\bT}{R_1}{X_0}\,,; 
            m_2\in \MatchGT{\bT}{R_2}{X_1}\}\\
            \cong\quad 
            B&=\{(\mu_{21},m_{21})\mid \mu_{21}\in \MatchGT{\bT}{R_2}{R_1}\,,\; m_{21}\in \MatchGT{\bT}{R_{21}}{X_0}\}\,,
        \end{aligned}
        \end{equation}
    where $X_1=R_{1_{m_1}}(X_0)$ and $R_{21}=\compGT{\bT}{R_2}{\mu_{21}}{R_1}$ such that for each corresponding pair $(m_2,m_1)\in A$ and $(\mu_{21},m_{21})\in B$, it holds that
        \begin{equation}
            R_{21_{m_{21}}}(X_0) \cong
            R_{2_{m_2}}(R_{1_{m_1}}(X_0))\,.
        \end{equation}
\end{theorem}

\section{Proofs}

\subsection{Proof of Theorem~\ref{thm:canrep}}\label{app:proofCanrep}

The statement of the theorem is equivalent to the following two properties:
    \begin{equation*}
        \begin{aligned}
        (i)\;&  &      \canRep{\bT}{\delta(R_{\mIO})}&=Id_{End_{\bR}(\hat{\bfC})}\\
        (ii)\;& &\forall R_1,R_2\in\LinEq{\bfC}:\quad
        \canRep{\bT}{\delta(R_{2})}\canRep{\bT}{\delta(R_{1})}&=
        \canRep{\bT}{\rap{\bT}{\delta(R_{2})}{\delta(R_{1})}}\,.
        \end{aligned}
    \end{equation*}
    By linearity, it suffices to verify these properties on an arbitrary basis vector $\ket{X}\in\hat{\bfC}$. For $(i)$, it suffices to verify that
    \[
        \canRep{\bT}{\delta(R_{\mIO})}\ket{X}=\sum_{m\in\MatchGT{\bT}{R_{\mIO}}{X}}\ket{R_{\mIO_m}(X)}=\ket{X}\,.
    \]
    Property $(ii)$ is a consequence of Theorem~\ref{thm:concur} (the Concurrency Theorem):
    \begin{align*}
        \canRep{\bT}{\delta(R_{2})}\canRep{\bT}{\delta(R_{1})}\ket{X}&=
        \sum_{m_1\in\MatchGT{\bT}{R_{1}}{X}}
        \sum_{m_2\in\MatchGT{\bT}{R_{2}}{R_{1_{m_1}}(X)}}\ket{R_{2_{m_2}}(R_{1_{m_1}}(X))}\\
        &=
        \sum_{\mu\in\MatchGT{\bT}{R_2}{R_1}}
        \sum_{m_{21}\in\MatchGT{\bT}{R_{2_{\mu}1}}{X}}
        \ket{R_{2_{\mu}1_{m_{21}}}(X)}\,.
    \end{align*}

\subsection{Proof of Theorem~\ref{thm:CTMCs}}
\label{sec:CTMCproofsApp}

\paragraph{Ad 1.:} It suffices to verify that direct derivations along a rule $R$ of the relevant form occurring in the two types of observables from any object $X$ satisfy $R_m(X)\cong X$. But this follows directly from the respective definitions of direct derivations.

\paragraph{Ad 2. \& 3.:} It again suffices to verify these properties on basis elements $\ket{X}$ of $\hat{\bfC}$, and for generic $R\in\LinEq{\bfC}$. By definition,
\begin{equation}
\bra{}\canRep{\bT}{\delta(R)}\ket{X}=\sum_{m\in \MatchGT{\bT}{R}{X}}\underbrace{\braket{}{R_m(X)}}_{=1_{\bR}}
=\vert \MatchGT{\bT}{R}{X}\vert\,.
\end{equation}
In both cases of semantics, a candidate match of $R$ into $X$ must satisfy the application condition. In the DPO case, in addition the relevant pushout complement must exist. Combining these facts allows to verify the formulae for $\jcOp{.}$.

\paragraph{Ad~4.:} The proof is straightforward generalization of the corresponding statement for the case of rewriting rules without conditions~\cite{bp2019-ext,nbSqPO2019}. Following standard continuous-time Markov chain (CTMC) theory~\cite{norris}, one may verify that the linear operator $\cH$ has a strictly negative coefficient diagonal contribution $\jcOp{H}$, a non-negative coefficient off-diagonal contribution $H$, thus $\cH$ satisfies $\bra{}\cH=0$. Since in addition a given $X\in\obj{\bfC}_{\cong}$ may be rewritten via direct derivations along the rules of the transition set only in finitely many ways, in summary $\cH$ fulfills all requirements to qualify as a conservative and stable $Q$-matrix (i.e.\ an infinitesimal generator) of a CTMC (cf.\ \cite{nbSqPO2019} for further  details).

\section{Details on the symbolic solution to the observable average counts in Example~\ref{ex:ugModel}}\label{app:se}

The ODE system of Example~\ref{ex:ugModel} may be solved in closed form as follows:
\begin{equation}
\begin{aligned}
\langle O_{\bullet}\rangle(t)&=\tfrac{\nu_{+}}{\nu_{-}} \left(1-e^{-t \nu_{-}}\right)\\
\langle O_{\bullet\vert\bullet}\rangle(t)&=\tfrac{\nu_{+}^2 e^{-\alpha  t} }{2 \alpha  \beta  \lambda  \nu_{-}^2}\left(\alpha  \beta  \varepsilon_{-} e^{\lambda  t}+2 \varepsilon_{+} \nu_{-}^2-2 \alpha  \kappa  \lambda  e^{\beta  t}+\beta  \lambda  \omega  e^{\alpha  t}\right)\\
\langle O_{\bullet\!-\!\bullet}\rangle(t)&=\tfrac{\varepsilon_{+} \nu_{+}^2 e^{-\alpha  t}}{2 \alpha  \beta  \lambda  \nu_{-}^2}\left(\alpha  \beta  \
e^{\lambda  t}-2 \alpha  \lambda  e^{\beta  t}+\beta  \lambda  \
e^{\alpha  t}-2 \nu_{-}^2\right)\\
\alpha&=\varepsilon_{-}+\varepsilon_{+}+2 \nu_ {-}\,,\;
 \beta=\varepsilon_{-}+\varepsilon_{+}+\nu_ {-}\\
\kappa&=\varepsilon_{-}+\nu_ {-}\,,\;
\lambda=\varepsilon_{-}+\varepsilon_{+}\,,\;
\omega=\varepsilon_{-}+2 \nu_ {-}\,.
\end{aligned}
\end{equation}
In particular, one may provide asymptotic formulae for $t\to\infty$:
\begin{equation}
\begin{aligned}
\langle O_{\bullet}\rangle(t)&\xrightarrow{t\to\infty}\tfrac{\nu_{+}}{\nu_{-}}\\
\langle O_{\bullet\vert\bullet}\rangle(t)&\xrightarrow{t\to\infty}
\tfrac{\nu_ {+}^2 (\varepsilon_{-}+2 \nu_{-})}{2 \nu_ {-}^2 (\varepsilon_{-}+\varepsilon_{+}+2 \nu_{-})}\\
\langle O_{\bullet\!-\!\bullet}\rangle(t)&\xrightarrow{t\to\infty}\tfrac{\varepsilon_{+} \nu_ {+}^2}{2 \nu_ {-}^2 (\varepsilon_{-}+\varepsilon_{+}+2 \nu_{-})}\,.
\end{aligned}
\end{equation}

\section{Technical details of typesetting the \MOD{} example}\label{app:ME}

For the interested readers, the following code may be used in either a standalone instance or via the \href{https://cheminf.imada.sdu.dk/mod/}{live playground} of \MOD{}~\cite{Andersen_2016} in order to reproduce the graphics for the Meisenheimer-2-3-rearrangement example of a organo-chemical reaction given in the main text. Note that since \MOD{} employs the traditional ``left-to-right'' convention for rules, the input and output patterns are given as ``right'' and ``left'', respectively.
\begin{python}
# Meisenheimer-2-3 rearrangement:
meisenheimer = ruleGMLString("""rule [
	left [
		edge [ source 1 target 2 label "-" ]
		edge [ source 2 target 3 label "=" ]
		edge [ source 1 target 5 label "-" ]
		node [ id 4 label "N" ]
		node [ id 5 label "O" ]
	]
	context [
		node [ id 1 label "C" ]
		node [ id 2 label "C" ]
		node [ id 3 label "C" ]
		edge [ source 4 target 5 label "-" ]
	]
	right [
		edge [ source 1 target 2 label "=" ]
		edge [ source 2 target 3 label "-" ]
		edge [ source 3 target 4 label "-" ]
		node [ id 4 label "N+" ]
		node [ id 5 label "O-" ]
	]
]""")
# Printing of the rule:
meisenheimer.print()
\end{python}

\end{document}

%% file: NBJK-RTLS-preamble.tex
\usepackage{graphicx}
\usepackage{subfigure}
\usepackage{xcolor}
\usepackage{amsmath,amssymb}
\usepackage{mathtools}
\usepackage{wrapfig}
\usepackage{longtable}
\usepackage{booktabs}
\usepackage{pythonhighlight}
\usepackage{array}
\usepackage{fontawesome}

\makeatletter
\def\bstr{b}
\def\bfstr{bf}
\def\cstr{c}
\def\fstr{f}
\def\strLst{A,B,C,D,d,E,F,G,H,I,J,K,L,M,N,O,P,Q,R,S,T,U,V,W,X,Y,Z}

\newcommand{\MkB}[1]{\expandafter\def\csname\bstr#1\endcsname{\mathbb{#1}}}
  \@for\i:=\strLst\do{%
    \expandafter\MkB \i     }
    
\newcommand{\MkBF}[1]{\expandafter\def\csname\bfstr#1\endcsname{\mathbf{#1}}}
  \@for\i:=\strLst\do{%
    \expandafter\MkBF \i     }

\newcommand{\MkCal}[1]{\expandafter\def\csname\cstr#1\endcsname{\mathcal{#1}}}
  \@for\i:=\strLst\do{%
    \expandafter\MkCal \i     }
    
\newcommand{\MkFrak}[1]{\expandafter\def\csname\fstr#1\endcsname{\mathfrak{#1}}}
  \@for\i:=\strLst\do{%
    \expandafter\MkFrak \i     }
    
\makeatother

\newcommand{\Lin}[1]{\mathop{\mathsf{Lin}}(#1)}

\newcommand{\LinAc}[1]{\overline{\mathsf{Lin}}(#1)}

\newcommand{\LinEq}[1]{\overline{\mathsf{Lin}}(#1)_{\sim}}

\newcommand{\ac}[1]{\mathsf{#1}} 

\newcommand{\Shift}{\mathsf{Shift}}

\newcommand{\Trans}{\mathsf{Trans}}

\newcommand{\MatchGT}[3]{\mathsf{M}^{{\text{\tiny $#1$}}}_{#2}(#3)}

\newcommand{\RMatchGT}[3]{\mathcal{M}^{{\text{\tiny $#1$}}}_{#2}(#3)}

\newcommand{\obs}[1]{\mathsf{obs}(#1)}

\newcommand{\pB}[1]{\mathsf{PB}(#1)}

\newcommand{\mono}[1]{\mathsf{mono}(#1)}

\newcommand{\epi}[1]{\mathsf{epi}(#1)}

\newcommand{\mor}[1]{\mathsf{mor}(#1)}

\newcommand{\iso}[1]{\mathsf{iso}(#1)}

\newcommand{\obj}[1]{\mathsf{obj}(#1)}

\newcommand{\mIO}{\mathop{\varnothing}}

\newcommand{\canRep}[2]{\overline{\rho}^{#1}_{\bfC}\left(#2\right)}

\newcommand{\bra}[1]{\left\langle #1\right\vert}
\newcommand{\ket}[1]{\left\vert #1\right\rangle}
\newcommand{\braket}[2]{\left\langle \left. #1 \right\vert #2\right\rangle}

\newcommand{\MOD}{\textsc{M\O{}D}}
\newcommand{\KAP}{\textsc{Kappa}}

\newcommand{\compGT}[4]{#2 {}^{#3}\!{\triangleleft}_{#1} #4}

\newcommand{\rap}[3]{#2 \star_{#1}{#3}}

\renewcommand{\vec}[1]{\underline{#1}}

\newcommand{\cond}[1]{\mathsf{cond}(#1)}

\newcommand{\jcOp}[1]{\hat{\bO}(#1)}

\newcommand{\Prob}[1]{\mathsf{Prob}(#1)}

\spnewtheorem{assumption}{Assumption}{\bfseries}{\itshape}

\newcommand{\inputtikz}[1]{%
 \vcenter{\hbox{\includegraphics{diagrams/#1.pdf}}}%
}

\colorlet{h1color}{blue!70!black} 
\colorlet{h2color}{orange!90!black} 
\colorlet{h3color}{blue!40!white} 
\colorlet{h4color}{green!40!black} 